\documentclass[11pt, letterpaper]{amsart}
\usepackage{graphicx, amssymb}

\newtheorem{thm}{Theorem}[section]

\newtheorem{claim}[thm]{Claim}
\newtheorem{lem}[thm]{Lemma}
\newtheorem{prop}[thm]{Proposition}
\theoremstyle{definition}
\newtheorem{defn}[thm]{Definition}
\theoremstyle{remark}
\newtheorem{rem}[thm]{Remark}

\numberwithin{equation}{section}
\newcommand{\norm}[1]{\left\Vert#1\right\Vert}

\newcommand{\set}[1]{\left\{#1\right\}}
\newcommand{\Real}{\mathbb R}
\newcommand{\Natural}{\mathbb N}
\newcommand{\eps}{\varepsilon}

\newcommand{\such}{\ | \ }
\newcommand{\prob}{\mathbb{P}}

\newcommand{\qprob}{\mathbb{Q}}
\newcommand{\expec}{\mathbb{E}}
\newcommand{\indic}{\mathbb{I}}

\newcommand{\filtration}{\mathbf{F} = \pare{\mathcal{F}_t}_{t \in [0, T]}}
\newcommand{\F}{\mathcal{F}}
\newcommand{\X}{\mathcal{X}}
\newcommand{\Y}{\mathcal{Y}}

\newcommand{\ud}{\mathrm d}
\newcommand{\inner}[2]{\left \langle #1 , #2 \right \rangle}

\newcommand{\conv}{\mathrm{conv}}

\newcommand{\pare}[1]{\left(#1\right)}

\newcommand{\dbra}[1]{[\kern-0.15em[ #1 ]\kern-0.15em]}
\newcommand{\dbraco}[1]{[\kern-0.15em[ #1 [\kern-0.15em[}
\newcommand{\K}{\mathcal{K}}
\newcommand{\eL}{\mathcal{L}}
\newcommand{\prices}{\mathcal{P}}
\newcommand{\D}{\mathcal{D}}
\newcommand{\Q}{\mathcal{Q}}
\newcommand{\B}{\mathcal{B}}
\newcommand{\Lb}{\mathbf{L}}
\newcommand{\oV}{\overline{V}}
\newcommand{\tV}{\tilde{V}}

\newcommand{\nemph}[1]{#1}
\newcommand{\seq}[1]{( #1_{n} )_{n\in\Natural}}

\newcommand{\toe}{\stackrel{e}{\to}}
\newcommand{\tog}{\stackrel{g}{\to}}

\begin{document}

\title[Stability of Utility Maximization with Random
Endowment]{Stability of the Utility Maximization problem with random
endowment in incomplete markets}
\author{Constantinos Kardaras}%
\address{Constantinos Kardaras, Mathematics and Statistics Department,
  Boston University, 111 Cummington Street, Boston, MA 02215, USA.}%
\email{kardaras@bu.edu}%

\author{Gordan {\v{Z}}itkovi{\'{c}}}%
\address{Gordan {\v{Z}}itkovi{\'{c}}, Department of Mathematics,
  University of Texas at Austin, 1 University Station, C1200, Austin,
  TX 78712, USA}%
\email{gordanz@math.utexas.edu}%


\date{\today}%

\maketitle

%
%
%
%
\begin{abstract}
We perform a stability analysis for the
utility maximization problem in a general semimartingale model where
both liquid and illiquid assets (random endowments) are present. Small
misspecifications of preferences (as modeled via expected utility), as well as views of the world or the market model (as modeled via
subjective probabilities) are considered. Simple sufficient conditions
are given for the problem to be well-posed, in the sense the optimal
wealth and the marginal utility-based prices are continuous functionals of preferences and probabilistic views.
\end{abstract}

\setcounter{section}{-1}


\section{Introduction}

\subsection{Expected Utility Maximization}
A mathematically sound, aesthetically pleasing and computationally
tractable description of optimal behavior of rational
economic agents under uncertainty comes from the \textsl{expected
  utility theory}: given a random outcome $X$ (e.g., a terminal
wealth, or a consumption stream) an agent's numerical
assessment of the ``satisfaction'' that $X$ provides is given by $\expec
[U(X)]$, where $U$ is a real-valued function, and
$\expec$ is the expectation corresponding to either
a \emph{physically} estimated
probability measure, or, in other circumstances, the
\emph{subjective} agent's view of the world.

Despite the criticism it received, expected utility theory has grown
widely popular and successful, mainly because it delivers
\emph{quantitative} results and, in some cases, even closed-form
solutions.  Among the seminal contributions in this vein in the field
of mathematical finance we single out \cite{sam69} (dealing with a
simple discrete-time Markovian model) and  \cite{merton69} (where
the problem of optimal investment in a continuous-time Markovian
framework is explicitly solved).  A more general approach that avoids
Markovian assumptions for the asset-price processes is the so-called martingale method.  In complete financial markets, this
methodology was introduced in \cite{pliska86} and later developed in
\cite{coxhuang89}, \cite{coxhuang91} and \cite{kls87}. For incomplete
financial models and continuous-time diffusion models important early
progress was made in \cite{hepear91} and \cite{klsx91}. \cite{ks99}
and \cite{ks03} contain a very complete picture of the solution of the
problem of expected utility maximization from terminal wealth in a
general semimartingale incomplete model when the wealth process
remains positive.

\subsection{Stability Analysis}

With problems of existence and uniqueness of optimal investment
virtually settled (at least for utilities defined on the positive
real line), interest in the \emph{stability analysis} (for the
solution of the problem of expected utility maximization under
perturbations of various initial conditions) has recently developed.
The problem of convergence of prices of illiquid assets, when the
prices of the liquid assets converge, was tackled in
\cite{hubsch98}.  In \cite{jounap04}, the authors look at an It\^o
process model and a convergent sequence of utility functions (i.e.,
misspecifications of a ``true'' utility function).  Convergence of
utilities is also considered in \cite{caras05}, but in a general
discrete-time setting. Continuity (and smoothness) properties with
respect to perturbations in the initial wealth and the quantities of
the illiquid assets have been studied in \cite{KraSir06}. \cite{Lar06} deals with utility-function misspecifications in
continuous-time models with general continuous-path semimartingale price processes, and illustrates the theory with applications to
certain  widely-used models. A different viewpoint is taken in
\cite{larzit06}. Therein \emph{model-}, rather than utility
misspecifications are studied: the asset price process $S^\lambda$
---  a continuous-path semimartingale --- is indexed by its
\textsl{market-price-of-risk} $\lambda$ (the parameter which is the
source of model misspecification).

In view of the previously listed works, one can argue that there has
been no unified treatment of the problem of stability under
simultaneous perturbations of both the utility functions and the
probability measures under which the expectations are taken.  The
aim of the present paper is to give insight into this problem in a
general semimartingale model, where the economic agent is,
additionally, endowed with a random payoff (illiquid assets).
Moreover, rather than merely providing a common platform for most of
the existing results, we generalize them in several directions. A
simple sufficient (and, in some cases, ``very close'' to necessary)
condition for stability is given, and several illuminating examples
dealing with various special cases are provided. We remark that in
this paper we deal with utility functions defined only on the
positive real line, since the theory of utility maximization with
random endowments for this case has been thoroughly understood. It
would be interesting to pursue whether a treatment of stability for
utility functions defined on the whole real line is possible, in the
spirit of the recent developments of \cite{BiFritGraHur07}, but we
are not dealing with this case in the present work.  We also note that the results appearing here have qualitative nature and constitute a zeroth order approach to the problem. The next natural step would be a first-order study, quantifying the infinitesimal change of value functions, the optimal wealth, as well as utility indifference prices. This would be accomplished by a study of the differentiability of the latter outputs with respect to smooth changes of the preferences and the agent's subjective views. We leave this important task as a future research project.

The structure of the paper is simple. After this Introduction,
section 1 describes the problem and states the main result, while
all the proofs are given in section 2.

\section{Problem Formulation and Statement
of the Main Result}

\subsection{Description of the modeling framework}

We start with a brief reproduction of the set-up and notation
introduced in \cite{hugkra04}, where the authors are concerned with
the problem of utility maximization with random endowment in
incomplete semimartingale markets.

\subsubsection{The financial market}
Let $(\Omega, \F, \mathbf{F}, \prob)$ be filtered probability space,
where the filtration $\filtration$ satisfies the \textsl{usual
conditions} of right continuity and $\prob$-completeness. The time
horizon $T > 0$ is fixed and constant. This assumption is in place
for simplicity only --- $T$ could be replaced by a finite stopping
time, as is the case in \cite{hugkra04} upon which we base our
analysis.

We consider a financial market with $d$ \emph{liquid} assets,
modeled by stochastic processes $S=(S^i)_{i = 1, \ldots, d}$. There
is also a ``baseline'' asset $S^0$ which plays the role of a
num\'eraire --- this amounts to the standard assumption $S^0\equiv
1$. The process $S$ is assumed to be a locally bounded
$\Real^d$-valued semimartingale (see \cite{delsch94} for the
economic justification of this essentially necessary assumption). Finally, in relation to the notion of absence of arbitrage, we posit
the existence of at least one \textsl{equivalent martingale
measure}, i.e., a probability measure $\qprob \sim \prob$ that makes
(each component of) $S$ a local martingale (see  \cite{delsch94} and
\cite{delsch98} for more information).

\subsubsection{Investment opportunities}
An initial capital $x > 0$ and a choice of an investment strategy
$H$ (assumed to be $d$-dimensional, predictable and $S$-integrable)
result in a wealth process $X= X^{x,H} = x + H \cdot S$, where
``$\cdot$'' denotes \emph{vector} stochastic integration.  In order
to avoid so-called doubling strategies, we restrict the class of
investment strategies in a standard way: the wealth process
$X=X^{x,H}$ is called \textsl{admissible} if
\[\prob[X_t \geq 0, \ \forall \ 0 \leq t \leq T] = 1.\]
An admissible wealth process $X$
is called {\em maximal} if  for
each $X' \in \X$ with $\prob[X'_T \geq X_T] = 1$ and $X'_0=X_0$, we
necessarily have $X = X'$, a.s. The
class of admissible wealth processes (starting from the initial
wealth $X_0=x$) is denoted by $\X(x)$. The union $\bigcup_{x>0}
\X(x)$ is denoted by $\X$.

On the dual side,
we  define the class of \textsl{separating measures} by
\[
\Q := \{ \qprob \ | \ \qprob \sim \prob, \textrm{ and } X \text{ is
} \qprob \textrm{-supermartingale for all } X \in \X \}.
\]
Thanks to the assumptions of no-arbitrage and local boundedness,
$\Q$ coincides with the set of all equivalent (local) martingale
measures, and is, therefore, non-empty. For future use, we restate
the (already imposed) assumption of \textsl{No Free
  Lunch with Vanishing Risk} as
\begin{equation} \label{ass: NFLVR} \tag{NFLVR}
\Q \neq \emptyset.
\end{equation}

\subsubsection{Illiquid assets}
Together with the liquid (traded) assets $S$, we assume the existence of $N$
\emph{illiquid} assets whose values at time $T$ are represented by
random variables $f^1, \ldots, f^N$. We allow for the case $N = 0$,
in which all assets are liquid. From the outset, the agents hold
some positions in illiquid assets, but, due to their illiquidity,
 they are not able to trade in them
(until the time $T$, at which all $N$ of them mature).
The only regularity
assumption on the illiquid assets is that they
can be super- and sub- replicated using the traded assets
$S$; in other words,  we assume (with the convention that
$\sum_{j=1}^0 \cdot =0$)
\begin{equation} \label{ass: S-REP} \tag{S-REP}
\X' := \Big \{ X \in \X \such X_T \geq
\sum_{j=1}^N |f^j|,
\ \prob
\textrm{-a.s.} \Big \} \neq \emptyset.
\end{equation}

To avoid trivial technical complications, we
 assume that the illiquid assets $f^1,\dots,f^N$ are
non-redundant when $N \geq 1$, in the sense that no linear
combination $\sum_{k=1}^N \alpha_k f^k$ --- where not all of the
$\alpha_k$'s are zero --- is replicable in the sense that there
exists a wealth process $X \in \X$ such that both $X$ and $-X$ are
maximal and $X_T = \sum_{k=1}^N \alpha_k f^k$. It is a standard
result (see for example Lemma 7 in \cite{hugkra04}) that this is
equivalent to saying that the set of \textsl{arbitrage-free prices}
for $f$ defined as
\begin{equation} \label{ass: N-TRAD} \tag{N-TRAD} \prices(f) := \{
  (\expec^{\qprob} [f^1], \dots, \expec^{\qprob} [f^N]) \, | \, \qprob
  \in \Q \} \text{ is an open set when $N>0$.}
\end{equation}
If \eqref{ass: N-TRAD} did not hold, we could always retain a
minimal set of (linear combinations) of the illiquid claims, and
regard all the others merely as outcomes of trading strategies using
the liquid assets only. It should become clear that \eqref{ass:
N-TRAD} is not needed for the results of the paper to hold and this
is why it is not assumed in our main Theorem \ref{thm: main} below.

Under the assumptions (\ref{ass: NFLVR}) and (\ref{ass: S-REP}), the class
\[
\Q' := \{ \qprob \in \Q \ | \  X \textrm{ is a } \qprob
\textrm{-uniformly integrable martingale for some } X \in \X' \}.
\]
can be shown to be non-empty. This follows from the fact that $\X'$
contains at least one maximal element $X$. The existence of a measure
$\qprob \in \Q$ that makes this maximal wealth process a uniformly
integrable martingale was established in \cite{delsch97}. Assumption
\eqref{ass: S-REP} implies that $f^j \in \Lb^1(\qprob)$ for all
$j=1,\dots, N$, $\qprob \in \Q'$.

\subsubsection{Acceptability requirements}
The notion of acceptability, related to that of maximality
introduced above,  plays a natural role when non-bounded random endowment is
present, as is thoroughly explained in \cite{delsch97} and \cite{hugkra04}. We say that a process
$X=X^{x,H}=x+H\cdot S$, with $H$ predictable and $S$-integrable, is
\emph{acceptable}, if there exists a maximal wealth process
$\breve{X} \in \X$ such that $X+\breve{X} \in \X$. Acceptability requires the
shortfall of a trading strategy to be bounded by a maximal wealth
process, rather than a constant, as in the case of the
admissibility requirements.

\subsubsection{The utility-maximization problem}
Starting with  initial wealth $x$ and $q^j$
units of each of the non-traded assets $f^j$ in the portfolio,
an economic agent can
invest in the market and achieve any of the wealths in the collection
\[
\X (x, q) := \{ X \equiv x + H \cdot S \such X \textrm{ is
acceptable and } X_T + \langle q, f \rangle \geq 0 \}.
\]
where $q \equiv (q^1, \ldots, q^N)$ and $\langle \cdot , \cdot
\rangle$ denotes inner product in the Euclidean space $\Real^N$ (if
$N=0$, the variable $q$ is absent). The agent's goal is to choose $X
\in \X(x, q)$ in such a way as to maximize $\expec^\prob [U(X_T +
\langle q, f \rangle)]$, where $\expec^\prob$ is used to
denote expectation under $\prob$, and the utility $U$ 
is a function mapping $(0, \infty)$ into $\Real$, which is strictly
increasing and strictly concave, continuously differentiable and
satisfies the Inada conditions: $U'(0 +) = \infty$, $U'(\infty) =
0$. The above \textsl{utility maximization problem} is considered
for all $(x, q) \in \K$, where $\K$ is the \emph{interior} of the
convex cone $\{(x, q) \ | \ \X(x, q) \neq \emptyset \}\subseteq
\Real^{N+1}$. In the liquid case $N=0$, (\ref{ass: NFLVR}) implies
that $\K=(0,\infty)=\mathrm{Int} [0,\infty)$. In the general case,
its geometry  depends on the interplay of the liquid and illiquid
assets. In Lemma 1 of \cite{hugkra04} it is shown that the assumption
\eqref{ass: S-REP} of sub- and super- replicability of the illiquid
assets is equivalent to $(x, 0) \in \K$, for all $x
> 0$ (and always, trivially, satisfied when $N=0$).

It is useful to consider the \textsl{value function} or
\textsl{indirect utility} of this problem as a function of both the
initial wealth $x$ that can be distributed in the liquid assets, and
the positions $q\in\Real^N$ held in the illiquid assets; thus, we
define the \textsl{indirect utility}
\begin{equation} \label{eq: primal value function}
u(x, q) := \sup_{X \in \X(x,q)} \expec^\prob [U(X_T + \langle q, f
\rangle)]
\end{equation}
for $(x, q) \in \K$. The specification of the indirect utility as a
function of both the initial capital and the holdings in the
illiquid assets is convenient if one wants to introduce
utility-based prices.

\subsubsection{Marginal utility-based prices}
For an  agent with an initial wealth $x$ and an initial position $q$
in $N\geq 1$ illiquid assets, a \textsl{marginal utility-based
price} for $f=(f^1,\dots,f^N)$ is a vector $p \equiv p(f; x,
q)\in\Real^N$ such that if $f$ were liquid and traded at prices $p$,
the utility-maximizing agent would be indifferent to changing
his/her positions in $f$. In more concrete terms, we must have $u(x,
q) \ \geq \ u(\tilde{x}, \tilde{q})$, for all $(\tilde x, \tilde q)
\in \K$ with $x +  \inner{q}{p} = \tilde{x} + \inner{\tilde q }{p}$.
In \cite{HugKraSch05}, the authors have shown that marginal
utility-based prices always exist, but do not, surprisingly, have to
be unique. More precisely, the {\em set} of marginal utility-based
prices for $f$ (with initial positions $x$ and $q$) is
\begin{equation} \label{eq: util indif prices}
\prices (f; x,q; U) \ := \ \{y^{-1} r \ | \ (y, r) \in
\partial u(x, q) \},
\end{equation}
where $\partial u(x, q)$ is the superdifferential of the concave
function $u$ at $(x, q) \in \K$.

\subsubsection{The dual problem}
In order to solve the \textsl{primal} (utility maximization)
problem,
it is useful to consider the related \textsl{dual problem}
\begin{equation} \label{eq: dual value function}
v(y, r) := \inf_{Y \in \Y(y, r)} \expec^\prob [V(Y_T)],
\end{equation}
where $V(y):= \sup_{x > 0} \{ U(x) - xy \}$ is the Legendre-Fenchel
transform of $U(\cdot)$ and $\Y(y, r)$ is defined to be the class of
all non-negative c\` adl\` ag processes $Y$ such that $Y_0 = y$,
$YX$ is a supermartingale for all $X \in \X$ and such that
$\expec[Y_T (X_T + \langle q, f \rangle)] \leq x y + \langle q, r
\rangle$ holds for all $(x, q) \in \K$ and $X \in \X(x, q)$. The
obvious simplifications apply when $N=0$. The dual problem
\eqref{eq: dual value function} is defined for all $(y,r) \in \eL$,
where we set $\eL := \mathrm{ri} (-\K)^{\circ}$, with
$(-\K)^{\circ}=\left\{ (y, r) \in \Real^{N+1} \ | \ xy + \langle q,
r \rangle\ \geq 0, \ \forall\, (x, q) \in \K \right\}$. In words,
$\eL$ is the \textsl{relative interior} of the \textsl{polar cone}
$(-\K)^{\circ}$ of $-\K$. We have the set equality
\begin{equation} \label{eq: arbitrage-free prices }
\prices(f) = \{ p \in \Real^N \such (1, p) \in \eL \}
\end{equation}
(see equation (9), p.~850 in \cite{hugkra04}) with $\prices(f)$
defined in \eqref{ass: N-TRAD} to be the set of
\textsl{arbitrage-free} prices for $f$. For future reference, for
any $p \in \prices(f)$ we set
\[
\Q'(p) := \{ \qprob \in \Q' \ | \ \expec^\qprob [f] = p \},
\]
where $\expec^\qprob [f]:=(\expec^\qprob [f^1], \dots, \expec^\qprob
[f^N])\in \Real^N$ and $\Q'(p) = \Q'$ if $N=0$.
The authors of \cite{hugkra04} show that $\Q'(p)
\neq \emptyset$ for \emph{all} $p \in \prices$.

\subsubsection{A theorem of Huggonier and Kramkov}
We conclude this section by stating a version of the main theorem
of \cite{hugkra04}, which  will be referred to throughout the sequel.

\begin{thm}[Huggonier and Kramkov (2004)] \label{thm: hugkra}
Suppose that $v(y,0) < \infty$ for all $y > 0$. Then, the functions
$u$ and $v$ are finitely valued on $\K$ and $\eL$, respectively, and are
conjugate to each other:
\[
v(y, r) = \sup_{(x, q) \in \K} \{ u(x, q) - xy - \langle q, r
\rangle \},
\]
\[
u(x, q) = \inf_{(y, r) \in \eL} \{ v(y, r) + xy + \langle q, r
\rangle \}.
\]
Furthermore, for each $(x,q)\in\K$ we have $\partial u (x, q)
\subseteq \eL$; actually,
\begin{equation}%
\label{equ:form-of-superdifferential}
    \begin{split}
  \partial u(x,q)=\begin{cases}
\{y\}\times R,& N\geq 1, \\ \{y\},& N=0
\end{cases}
\end{split}
\end{equation}
for some $y=y(x,q)\in (0,\infty)$ and some compact and convex
set $R=R(x,q)\subseteq\Real^N$.
The
optimal solutions $\hat{X}(x, q)$ and $\hat{Y}(y, r)$ for the primal and
dual problems exist for all $(x, q) \in \K$ and $(y, r) \in
\eL$. Moreover,
if $(y, r) \in \partial u(x, q)$, we have the $\prob$-a.s.~equality  $\hat{Y}_T (y, r) = U'(\hat{X}_T (x, q) + \langle q, f
\rangle)$

\end{thm}

\begin{rem}
For all $(x,q)\in\K$, the superdifferential $\partial u (x, q)$ is a
compact and convex subset of $\eL$. Moreover, since $y>0$ for
$(y,r)\in\partial u(x,q)$, equation \eqref{eq: util indif prices}
and the set-equality \eqref{eq: arbitrage-free prices } imply that
$\prices (f; x,q; U)$ is a convex and compact subset of $\prices
(f)$ --- in other words, marginal utility-based prices are
arbitrage-free prices.
\end{rem}

We note some further properties of the utility and value functions
above. It follows from the properties of  convex conjugation that
the function $V$ is strictly convex, continuously differentiable and
strictly decreasing on its natural domain. The value functions $u$
and $v$ are conjugates of each other
--- $u(\cdot,q)$ is strictly concave, strictly increasing, while
$v(\cdot,q)$ is strictly convex and strictly decreasing. Both
$u(\cdot,q)$ and $v(\cdot,q)$ are continuously differentiable.
Considered as functions of the second argument $u(x,\cdot)$ and
$v(y,\cdot)$ are continuous on the interiors of effective domains.

\subsection{Stability analysis}

Having described the utility-maximization setting of \cite{hugkra04},
we turn to the central question of the present paper: \emph{what are
  the consequences of model and/or preference misspecification for
   the optimal investment problem (as described in the previous
  section)?
}

\subsubsection{Problem formulation}
In mathematical terms, we can ask whether the mapping that takes as
inputs a utility function $U$ and a probability measure $\prob$ and
produces the optimal wealth process and the set of utility-based
prices for contingent claims (the illiquid assets) is continuous. Of
course, appropriate topologies on the sets of
the probability measures, utility functions, terminal wealth processes
and prices need to be introduced.

Focusing on the special case of the logarithmic utility in a complete
It\^o-process market, the authors of \cite{larzit06} determine
certain   conditions on
topologies governing the convergence of stock-price processes, which are
\emph{necessary} for convergence in probability on the space of the
terminal wealth processes in \emph{all} models. A similar
approach in our case obviates the need for, at least, the
following set of assumptions:
\begin{itemize}
\item[i)] the class of probability measures is endowed with the
  topology of convergence in \emph{total variation}, and
  \item[ii)] the space of utility functions is topologized by \emph{pointwise}
    convergence.
\end{itemize}

\begin{rem}\label{rem:modes}
In general, the topology of pointwise convergence lacks the
operational property of metrizability. However, when restricted to a
class of concave functions --- such as utility functions ---  it
becomes equivalent to the metrizable topology of uniform convergence
on compact sets. From the economic point of view, such convergence
is natural because --- despite its apparent coarseness ---
 it implies pointwise (and locally uniform) convergence
of derivatives (marginal utilities), and thus, convergence in the
local Sobolev space $W^{1,\infty}_{loc}$. It is implied, for
example, by the convergence of (absolute or relative) risk aversions
under the appropriate normalization. The pointwise convergence of
utility functions is the most used notion of convergence for utility
functions in economic literature (see \cite{jounap04} or
\cite{caras05} in the financial framework, or \cite{Bac83} for a
more general discussion and relation to other, less used notions of
convergence).
\end{rem}

In the sequel, we consider two sequences $(\prob_n)_{n \in
\Natural}$ and $(U_n)_{n \in \Natural}$ of probability measures and
utilities, together with the
``limiting'' probability measure $\prob$
and utility function $U$. These will  always be assumed to satisfy the
following (equivalency and) convergence condition:
\begin{equation} \label{ass: CONV} \tag{CONV}
\forall\, n \in \Natural,\,\, \prob_n \sim \prob,\  \lim_{n \to \infty}
\prob_n = \prob \textrm{ in total variation and} \ \lim_{n \to \infty}
U_n = U \textrm{ pointwise}.
\end{equation}

\begin{rem}
Some aspects of the approach of
\cite{larzit06} can be recovered in our setting
when the utility function $U$ is kept constant and there are no
illiquid assets ($N=0$). In order to see that, recall that in
\cite{larzit06}, the authors consider a general (right-continuous
and complete) filtration $\filtration$, on which a one-dimensional
\nemph{continuous} local martingale $M$ is defined. They vary the
model by considering a sequence $(\lambda_n)_{n\in\Natural}$ of
{market-price-of-risk} processes, giving rise to a sequence of
stock-price processes
\begin{equation}
\nonumber
\begin{split}
      \ud S^{\lambda_n}(t)= \lambda_n(t)\, \ud \langle M \rangle(t)+ \ud M(t).
\end{split}
\end{equation}
They study the convergence of the outputs of the
utility-maximization problems in the sequence $(S^{\lambda_n}
)_{n\in\Natural}$ of models, while keeping the ``physical'' measure
$\prob$ fixed.

In our framework, we keep the functional representation of
the models constant (as the same functions mapping
$\Omega$ into the appropriate co-domain), but the measure $\prob$
changes. To see the connection, let $S$ be a continuous-path semimartingale. Then, $\ud S(t) = \lambda (t)\, \ud \langle M  \rangle(t)+ \ud M (t)$, where $M$ is a local $\prob$-martingale. For $n \in \Natural$ let $\prob_n \sim \prob$; then, Girsanov's theorem enables us to write $\ud S(t) = \lambda_n(t)\, \ud \langle M \rangle(t)+ \ud M_n(t)$, where $M_n$ is a local $\prob_n$-martingale with $\langle M_n \rangle = \langle M \rangle$. It is straightforward to check that $\lim_{n \to \infty}
\prob_n = \prob$ in total variation implies $\lim_{n \to \infty} \int_0^T \|
\lambda_n(t) - \lambda(t)  \|^2 \ud \langle M \rangle (t) = 0$ in $\Lb^0$. Conversely, the latter convergence, coupled with requiring that $M$ has the predictable representation property with respect to the filtration $\mathbf{F}$ and some uniform integrability conditions, imply that $\lim_{n \to \infty}
\prob_n = \prob$ in total variation.
\end{rem}

\begin{rem}

The equivalence of all probability measures $(\prob_n)_{n \in
\Natural}$ to $\prob$ as required by \eqref{ass: CONV} is a rather
strong condition --- in particular, it pins down the quadratic
variation of $S$ and this means that model misspecifications with
respect to volatility in simple It\^o-process models cannot be
dealt.  Our choice to impose such a requirement nevertheless is
based on the following two observations:
\begin{enumerate}
\item Stability in the general (non-equivalent) case can only be
  studied in the distributional sense; equivalence allows one to
  talk about convergence in probability. Such problems do not arise
  when one only considers numerical objects, such as  prices of
  contingent claims for example.
\item The structure of the dual sets (the sets of equivalent
  martingale measures) in the
  limit and that in the pre-limit models differ greatly in
  typical non-equivalent cases. This puts a severe limitation on the
  applicability of our method.
\end{enumerate}
In special cases, however, there exists a simple way
around the equivalence assumption, based on the
observation that the subject of
importance is not the asset-price vector $S$ itself, but the
collection of all wealth processes that are to be used in the
utility maximization problem. This simple observation allows, for example, treatment of stochastic volatility models. The example below illustrates the
general principle of how the equivalence requirement can be
``avoided'':

Suppose that under $\prob$ we have the dynamics $\ud S^i_t / S^i_t
=: \ud R^i_t = \mu^i_t \ud t + \sum_{j = 1}^d \sigma_t^{i j} \ud
W^j$ for $i = 1, \ldots, d$, where $W = (W^j)_{1 \leq j \leq d}$ is
an $\mathbf{F}$-Brownian motion, $\mu = (\mu^i)_{1 \leq i \leq d}$
and $\sigma = (\sigma^{ij})_{1 \leq i \leq d, 1 \leq j \leq d}$ are
$\mathbf{F}$-predictable and $\sigma$ is assumed to be
non-singular-valued. It follows that for the returns vector $R =
(R^i)_{1 \leq i \leq d}$ we can write $\ud R_t = \sigma_t (\lambda_t
\ud t + \ud W_t)$, where $\lambda := \sigma^{-1} \mu$ is the Sharpe
ratio. The non-singularity of $\sigma$ implies that the set of wealth processes obtained by trading in $S$ is the same as the one obtained by trading in assets with returns given by $\widetilde{R} =
(\widetilde{R}^i)_{1 \leq i \leq d}$ satisfying $\ud \widetilde{R}_t = \lambda_t
\ud t + \ud W_t$. This trick allows to get rid of the dependance on $\sigma$.

Suppose now we want to check the effect of changing both
$\sigma$ and $\mu$
--- for example we want to see what will happen if $(\mu^{(n)},
\sigma^{(n)})$ converge to $(\mu, \sigma)$ in some sense. Define
$\lambda^{(n)} = (\sigma^{(n)})^{-1} \mu^{(n)}$ (assume that each
$\sigma^{(n)}$ is non-singular-valued) and $\prob_n$ via the density
(assuming that the exponential local martingale below is uniformly
integrable):
\[
\frac{\ud \prob_n}{\ud \prob} \Big |_{\F_T} = \exp \Big( \int_0^T
(\lambda^{(n)}_t - \lambda_t ) \ud W_t - \frac{1}{2} \int_0^T \|
\lambda^{(n)}_t - \lambda_t  \|^2 \ud t  \Big)
\]
As long as $\lim_{n \to \infty} \int_0^T \| \lambda^{(n)}_t -
\lambda_t  \|^2 \ud t = 0$ (in probability) we have $\lim_{n \to
\infty} \prob_n = \prob$ in total variation. Define new return processes $R^{(n)}$ via $\ud R^{(n)}_t  =
\sigma_t^{(n)} \big( \lambda^{(n)}_t \ud t + \ud W^{(n)}_t \big) = \mu^{(n)} \ud t + \sigma_t^{(n)} \ud W^{(n)}_t$, where $W^{(n)} := W -
\int_0^\cdot (\lambda^{(n)}_t - \lambda_t ) \ud t$ is
$\prob_n$-Brownian motion. The induced set of wealth
processes by investing in asset-prices with returns $R^{(n)}$ is the same as the one obtained if the asset-prices had returns $\widetilde{R}^{(n)}$ that satisfied $\ud \widetilde{R}^{(n)}_t  =
\lambda^{(n)}_t \ud t + \ud W^{(n)}_t = \ud \widetilde{R}_t$, which in turn is the same as the \emph{original} set of wealth processes obtained by investing in $S$. In this indirect way, we can study changes of both drift and volatility in the model, while keeping our framework of only changing the probability measure and not the asset prices.
\end{rem}

\subsubsection{A uniform-integrability condition}
Unfortunately, the modes of convergence in \eqref{ass: CONV} are not
strong enough for stability: \cite{larzit06} contains  a simple
example. In the setting of their example, $T=1$ and there exists one
liquid asset $S$ whose $\prob$-dynamics ($\prob$ being the
``limiting measure'') is given by $\ud S_t = S_t \ud W_t$. $W$ is a
$\prob$-Brownian motion, and the filtration is the (augmentation of
the) one generated by $W$. The sequence $(\prob_n)_{n\in\Natural}$,
of measures is defined via $\ud \prob_n / \ud \prob = \varphi_n
(W_1)$, where $(\varphi_n)_{n \in
  \Natural}$ is a sequence of positive real functions
with $\lim_{n \to \infty} \varphi_n = 1$, pointwise. The utility function involved --- in their treatment only the model
changes and the utility is fixed  --- is \emph{unbounded from above}
(and is, in fact, a simple power function). What the authors of
\cite{larzit06} show is that convergence of the optimal wealth
processes in probability might fail --- convergence of $(\prob_n)_{n
\in \Natural}$ to $\prob$ in total variation is simply not enough.
Moreover, their choice of the functions $\varphi_n$ is such that
$\ud \prob_n/\ud \prob\to 1$ in $\Lb^2$, and a simple variation of
their argument may be used to show that, in fact, the $\Lb^p$
convergence will not be universally sufficient, no matter how large
$p\in (1,\infty)$ is chosen. The appropriate strengthening of the
requirement (\ref{ass: CONV}), as shown by \cite{larzit06}, is  the
replacement of the classical $\Lb^p$ spaces by the Orlicz spaces
related to the utility function $U$. In the present setting, where
the variation in the model, as well as in the utility function, has
to be taken into account, such a replacement leads to the following
condition (in which $V_n^+ (x) := \max \{ V_n(x), 0 \}$):
\begin{equation} \label{ass: UI} \tag{UI}
\forall\, p\in\prices,\ \exists\, \qprob\in \Q'(p),\ \forall\,
y>0,\ \ \left( \frac{\ud \prob_n}{\ud \prob} \ V_n^+ \Big( y \frac{\ud
\qprob}{\ud \prob_n} \Big) \right)_{n \in \Natural}\text{ is
$\prob$-uniformly integrable.}
\end{equation}

\subsubsection{On condition (\ref{ass: UI})}
The following special cases  illustrate the meaning and restrictiveness of
 the condition (\ref{ass: UI}).
 The convergence requirement
(\ref{ass: CONV}) is  assumed throughout.
\begin{enumerate}
\item It has been shown in \cite{larzit06} that (the appropriate
  version of) the condition
  \eqref{ass: UI} is both sufficient \emph{and necessary} in
  complete financial markets. In the incomplete case, and still in
  the setting of \cite{larzit06}, it is ``close to'' being necessary ---
  the gap arising because of the technical issues stemming from the
  fact that the dual minimizers do {\em not} have to be countably-additive measures.
\item When there are no
  illiquid assets ($N=0$), the set $\prices$ has no meaning and any
  martingale measure $\qprob$ can be used in (UI). Also, in the case when the market is complete, the set $\prices$ is a singleton and the unique equivalent martingale measure $\qprob$ has to be used in (UI).
\item
The \eqref{ass: UI}
condition  is immediately satisfied if the sequence $(U_n)_{n \in \Natural}$
is uniformly bounded from above. Indeed, in that case we have
$\sup_{n \in \Natural} V^+_n \leq C$ for some $C > 0$ (the uniform
upper bound on the utilities) and the sequence $(\ud \prob_n / \ud
\prob)_{n \in \Natural}$ is $\prob$-uniformly integrable in view of
its $\Lb^1 (\prob)$ convergence.
\item
If the previous example corresponds to the duality between $\Lb^{\infty}$ and
$\Lb^{1}$, the present one deals with the case of $\Lb^{\hat{p}}$ and $\Lb^{\hat{q}}$,
${\hat{p}}^{-1}+{\hat{q}}^{-1}=1$. Indeed, assume the  following
conditions:
\begin{enumerate}
\item there exist constants $c>0$,
$d\in\Real$ and
$0<\alpha<1$ (the case $\alpha=0$ corresponds to the logarithmic
function, and can be treated in a similar fashion)
such that
 $U_n(x) \leq c x^{\alpha} + d$, for all $n\in\Natural$,
\item the sequence
 $(\ud \prob_n/\ud \prob)_{n \in \Natural}$ is bounded in  $\Lb^{\hat{p}}$, for some
 ${\hat{p}}>(1-\alpha)^{-1}$, and
\item for each $p\in\prices$ there exists
$\qprob_p\in {\mathcal Q}'(p)$
  such that $(\ud \qprob_p / \ud \prob)^{-1} \in \Lb^{\hat{q}}$, where
  we set ${\hat{q}} := \tfrac{{\hat{p}}\alpha}{{\hat{p}}(1-\alpha)-1}$. Note
  that this requirement is not as strong as it seems, as it is closely
  related to the finiteness in the dual problem.
\end{enumerate}
 Then, for all $y>0$,
\begin{equation}%
\label{equ:estimate}
    \begin{split}
      V_n(y)=\sup_{x>0} \{ U_n(x)-xy \}\leq \sup_{x>0} [c x^{\alpha} +d -
      xy]\leq  C y^{- \tfrac{\alpha}{1 - \alpha}} + D,
    \end{split}
\end{equation}
where $C,D\in\Real$ are positive constants.
 For arbitrary but fixed $p\in\prices$ and $y>0$ define
$\gamma := {\hat{q}}{\hat{p}}(1-\alpha)({\hat{q}} +
{\hat{p}}\alpha)^{-1}$ so that $1<\gamma<{\hat{p}}$, where
${\hat{q}}>0$ has been defined above. H\"older's inequality (applied
in the last inequality below) and the estimate (\ref{equ:estimate})
imply that
\begin{eqnarray*}
  \expec\left[ \left( \frac{\ud \prob_n}{\ud \prob} \ V_n^+ \Big(
          y \frac{\ud \qprob_p}{\ud \prob_n} \Big)
        \right)^{\gamma}\right] &\leq& \expec \left[ \left( C \frac{\ud
            \prob_n}{\ud \prob} \left( y \frac{\ud \qprob_p}{\ud
              \prob_n } \right)^{- \frac{\alpha}{1 - \alpha}}+ D
          \frac{\ud \prob_n}{\ud \prob} \right)^{\gamma}
      \right] \\
   &\leq& 2^{\gamma-1}C^{\gamma} y^{-\tfrac{\alpha}{1-\alpha}} \expec \left[
        \left(\frac{\ud \prob_n}{\ud
            \prob}\right)^{\frac{\gamma}{1-\alpha}} \left( \frac{\ud
            \qprob_p}{\ud \prob } \right)^{- \frac{\gamma \alpha}{1 -
            \alpha}}\right] + 2^{\gamma-1}D^{\gamma} \expec\left[ \left( \frac{\ud
            \prob_n}{\ud \prob} \right)^{\gamma}
      \right] \\
   &\leq& 2^{\gamma-1}C^{\gamma} y^{-\tfrac{\alpha}{1-\alpha}} \expec \left[
        \left(\frac{\ud \prob_n}{\ud
            \prob}\right)^{{\hat{p}}}\right]^{\frac{\gamma}{{\hat{p}}(1-\alpha)}}\,
      \expec\left[\left( \frac{\ud \qprob_p}{\ud \prob }
        \right)^{{-\hat{q}}}\right]^{1-\frac{\gamma}{{\hat{p}}(1-\alpha)}} \\
   &+& 2^{\gamma-1}D^{\gamma} \expec\left[ \left( \frac{\ud \prob_n}{\ud
            \prob} \right)^{\gamma} \right],
\end{eqnarray*}
which implies (UI).
\item A family $(U_n)_{n \in
\Natural}$ of utility functions is said to have a \emph{uniform
reasonable asymptotic elasticity}, if there exist constants $x_0>0$
and $\delta<1$ such that $x U_n'(x)\leq \delta U_n(x)$ for all $x
> x_0$ and $n\in\Natural$. Then, one can show (see
Proposition 6.3 in \cite{KraSch99}) that for each fixed $y>0$, there
exist $k,l>0$ such that $V_n^+(y z)\leq k V^+_n(z)+l$ for all $z>0,
n\in\Natural$. In other words, under uniform reasonable asymptotic
elasticity the ``annoying'' universal quantification over all  $y >
0$ in (UI)
 can be left out --- considering only the case $y=1$ is enough.
\item Several other sufficient conditions for (UI) in the case when
  $V_n=V$ for all $n\in\Natural$ and $N=0$ are given in \cite{larzit06}.
\end{enumerate}

\subsubsection{The main result}
The statement of our main result, whose proof
will be the given in  Section \ref{sec: proofs} below, follows.
In order to keep the unified notation for the cases $N=0$ and
$N>0$, we introduce the following conventions (holding throughout the
remainder of the paper):
all the statements in the sequel will notationally correspond to the
case $N>0$, and should be construed literally in that case. When
$N=0$, the arguments $r$ should be understood to take values in the
one-element set $\Real^0$, which we identify with $\set{0}$. Similarly,
the variables $p$ and $q$ will take the value $0$, and ${\mathcal
  Q}(0)={\mathcal Q}'$. In this case, a pair such as $(x,0)$ will be
identified with the constant $x\in\Real$.

\begin{thm} \label{thm: main}
Assume that \eqref{ass: NFLVR} and \eqref{ass: S-REP} are in
force, and  consider a sequence $(\prob_n)_{n \in
\Natural}$ of probability measures and a sequence $(U_n)_{n \in
\Natural}$ of utility functions such that \eqref{ass: CONV} and
\eqref{ass: UI} hold. Furthermore, let
$(x_n, q_n)_{n \in \Natural}$ be a $\K$-valued sequence
with $\lim_{n \to \infty} (x_n, q_n)
=: (x, q) \in \K$ and $(y_n, r_n)_{n \in
\Natural}$ an $\eL$-valued sequence
with $\lim_{n \to \infty} (y_n, r_n) =: (y, r) \in \eL$.

Set $u_n= u (x_n, q_n; U_n, \prob_n)$, $u_{\infty}=u(x,q; U,
\prob)$, $v_n= v (y_n, r_n; U_n, \prob_n)$, $v_{\infty}=v(y,r; U,
\prob)$, and let $\frac{\partial}{\partial x} u_n$,
$\frac{\partial}{\partial x} u_{\infty}$, $\frac{\partial}{\partial
y} v_n$, $\frac{\partial}{\partial y} v_{\infty}$, be the
corresponding derivatives with respect to the first variable.
Similarly,  set $\hat{X}_n=\hat{X}_T (x_n, q_n; U_n, \prob_n)$,
$\hat{X}_{\infty} = \hat{X}_T (x, q; U, \prob)$, $\hat{Y}_n=
\hat{Y}_T (y_n, r_n; V_n, \prob_n)$ and $\hat{Y}_\infty =\hat{Y}_T (y, r; V,
\prob)$. Finally, set $\prices_n=\prices (f; x_n, q_n; U_n,
\prob_n)$ and $\prices_{\infty}= \prices (f; x, q; U, \prob)$.

Then, we have the following limiting relationships for the value functions
and the optimal solutions in the primal and dual problems:
\begin{enumerate}
 \item
$\lim_{n \to \infty} u_n=u_{\infty}$, $\lim_{n \to \infty}
v_n=v_{\infty}$,
$\lim_{n \to \infty}  \frac{\partial}{\partial x}u_n=
\frac{\partial}{\partial x}u_{\infty}$, and
$\lim_{n \to \infty}  \frac{\partial}{\partial y}v_n=
\frac{\partial}{\partial y}v_{\infty}$.
 \item $
\lim_{n \to \infty} \hat{X}_n=\hat{X}_{\infty} \text{ and }
   \lim_{n\to\infty} \hat{Y}_n = \hat{Y}_{\infty}
\text{ in
       $\Lb^0$,}$
where, as usual, $\Lb^0$ is the family of all random
       variables endowed with the topology of  convergence
       in probability.
 \item for all $\epsilon > 0$, there
   exists $n_0\in\Natural$ such that
\[\prices_n \subseteq \prices_{\infty}
 + \epsilon B^N,\text{ for }n\geq n_0,
 \]
 where $B^N$ is the open ball of unit radius in $\Real^N$.
\end{enumerate}
\end{thm}

\begin{rem}
The set-inclusion $\prices_n \subseteq
\prices_{\infty} + \epsilon B^N$ for all $n$ large
enough is an \textsl{upper hemicontinuity}-type property of the
correspondence of marginal utility-based prices. It says that all
possible limit points of all possible sequences of marginal
utility-based prices will belong to the limiting price-set. It does
\emph{not} imply that this last set will be equal to the set of all
these possible limit points
--- indeed, it might be strictly larger.
\end{rem}
\section{Proofs} \label{sec: proofs}

This section concentrates on the proof of our main Theorem \ref{thm:
  main}.  First, we prove a lower semicontinuity-type result for of
the dual value function, which, interestingly, does \emph{not} depend
on the assumption (\ref{ass: UI}) from subsection \ref{subsec: lower
  semicontinuity of dual}.  Then, in subsection \ref{subsec:
  continuity of dual}, we use both \eqref{ass: CONV} and \eqref{ass:
  UI} to establish a complementary upper semicontinuity-type property
for the dual value function. Continuity of the primal value function
and upper hemicontinuity of the correspondence of
marginal utility-based prices
are proved in subsection \ref{subsec: continuity of
  primal and util indif prices}. Finally, subsection \ref{subsec:
  continuity of optimal wealth} deals with convergence in
$\Lb^0$ of the dual optimal element.
Convergence in $\Lb^0$ of the optimal terminal wealths is then
established using the
continuity
of the value functions.

\subsection{Preliminary remarks} We start by making some
remarks and conventions that will be in force throughout
 the proof.  Since there are many different probability measures
floating around, we choose $\prob$ to serve as the baseline:
\emph{all} expectations $\expec$ in the sequel will be taken with
respect to the probability $\prob$ --- we then consider the
Radon-Nikodym densities $Z_n := \ud \prob_n/ \ud \prob$ and use them
whenever we want to take expectation with respect to some $\prob_n$.
The space $\Lb^0$ of all a.s.-finite random variables is the same
for all (equivalent) probabilities and thus requires no identifier.
The notation $\Lb^1$ is reserved for $\Lb^1(\prob)$. Observe
that the convergence $\lim_{n \to \infty} \prob_n = \prob$ in total
variation is equivalent to the convergence $\lim_{n \to \infty} Z_n
= 1$, in $\Lb^1$. By Scheffe's Lemma (see \cite{Wil91}, p.~55)
 this is equivalent to the
(seemingly weaker)
statement $\lim_{n \to \infty} Z_n = 1$ in $\Lb^0$.

\smallskip

Let us move on to the discussion of utility functions. Note that
pointwise (and thus, by concavity,  uniform on compacts) convergence
of the sequence $(U_n)_{n \in \Natural}$ to a utility $U$ will imply
pointwise convergence of the sequence of Legendre-Fenchel transforms
$(V_n)_{n \in \Natural}$ to the corresponding Legendre-Fenchel
transform $V$ of the limiting utility $U$. We actually get a lot
more: the sequences $(U_n)_{n \in \Natural}$, $(V_n)_{n \in
\Natural}$ as well as their derivatives $(U'_n)_{n \in \Natural}$,
$(V'_n)_{n \in \Natural}$ converge \emph{uniformly} on compact
subsets of $(0, \infty)$ to their respective limits $U$, $V$, $U'$
and $V'$ (see, e.g., \cite{RocWet98} for a general statement or
\cite{larzit06} for a simple self-contained proof of this result). A
multidimensional version of this result will be used later on in
subsection \ref{subsec: continuity of optimal wealth}.

Also, without loss of generality we assume  that
each of the utility functions involved here is normalized in such a
way as to have $U_n (1) = 0$ and $U_n'(1) = 1$ --- this will mean
that $V_n(1) = V'_n(1) = -1$. One can check that nothing changes in
the validity of our Theorem \ref{thm: main} if we make this simple
affine transformation in the utilities, but the proofs below will be
much cleaner. Indeed,  we
can define a new sequence $(\hat{U}_n)_{n \in \Natural}$ via
\[
\hat{U}_n (x) = \frac{U_n(x) - U_n(1)}{U'_n(1)};
\]
pointwise convergence of the original sequence $(U_n)_{n \in
\Natural}$ implies pointwise convergence of both
$(\hat{U}_n)_{n \in
\Natural}$ and
 $(\hat{V}_n)_{n \in \Natural}$.

 \subsection{A lower semicontinuity-type property of the dual value function}
 \label{subsec: lower semicontinuity of dual}

We assume that \eqref{ass: NFLVR}, \eqref{ass: S-REP} and
\eqref{ass: CONV} hold throughout this subsection. The assumption
\eqref{ass: UI} is not yet needed.

\subsubsection{Preparatory work}
 Notice that  $(Z_n)_{n \in
\Natural}$ is $\prob$-uniformly integrable, and, more generally, that
the \textsl{convex hull} $\conv(Z_n ; n \in \Natural)$ is
$\prob$-uniformly integrable, as well.
Observe that
\begin{equation}%
\label{equ:dualvn}
    \begin{split}
v_n (y, r) \equiv v (y, r; \ U_n, \prob_n) = \inf_{g \in \D(y, r)}
\expec [Z_n V(g / Z_n)] = \expec [Z_n V(g_n / Z_n)],
    \end{split}
\end{equation}
where
\begin{equation} \label{eq: dual domain}
\D(y, r) := \{ g \in \Lb^0 \such 0 \leq g \leq Y_T \textrm{ for some
} Y \in \Y (y, r) \},
\end{equation}
$\Y (y, r)$ is the class of supermartingale deflators corresponding
to the limiting probability measure $\prob$, and $g_n \in \D(y, r)$
attains the infimum in (\ref{equ:dualvn}). We wish to show that
\begin{equation} \label{eq: lower semicontinuity of dual}
v(y,r) \ \leq \ \liminf_{n \to \infty} v_n(y, r)
\end{equation}

The ``liminf''  in \eqref{eq: lower
semicontinuity of dual} we can be safely regarded as an
\emph{actual} limit, passing to an attaining subsequence if necessary. By
the same token, we can also assume that the
convergence $\lim_{n \to \infty} Z_n = 1$  holds almost
surely, and not only in $\Lb^0$.

Lemma A.1.1 from \cite{delsch94} provides us with a finite random
variable $h\geq 0$, and a sequence $(h_n)_{n \in \Natural}$ such
that
\begin{equation} \label{eq: nice DS L0 convexity lemma}
\forall\, n\in\Natural,\ h_n \in \conv(g_n, g_{n+1}, \ldots)\text{
and }\lim_{n \to \infty} h_n = h,\text{ a.s.}
\end{equation}
In
\cite{hugkra04} the authors show
that for all $(y, r) \in \eL$ the convex set $\D(y,r)$ is closed in $\Lb^0$, thus we
have $h \in \D(y,r)$. For concreteness, let us  write $h_n = \sum_{k =
n}^{m_n} \alpha^n_k g_k$ for some $m_n \geq n$ and $0 \leq
\alpha^n_k \leq 1$ such that $\sum_{k = n}^{m_n} \alpha^n_k = 1$. We
then also set $\zeta_n := \sum_{k = n}^{m_n} \alpha^n_k Z_k$ and
observe that $\lim_{n \to \infty} \zeta_n = 1$ holds almost surely;
here, it is crucial that we have $\lim_{n \to \infty} Z_n = 1$
almost surely and not only in $\Lb^0$ --- $\Real$ is a locally
convex space, while $\Lb^0$ is not.

\subsubsection{On the sequence $(V_n)_{n \in \Natural}$}

For $n \in \Natural$ and $\epsilon \in (0,1)$, define the function
$V_n^\epsilon$ as follows: set $V_n^\epsilon (x) = V_n(x)$ for $x \geq
\epsilon$, and extend $V_n^\epsilon$ to $[0,\epsilon)$ in an affine and
continuously differentiable way, i.e.~ match the zeroth and first
derivatives at $x=\epsilon$.
This recipe uniquely determines a decreasing and convex function
$V_n^\epsilon$.
Of course, $\lim_{\epsilon
\downarrow 0} \uparrow V_n^\epsilon (x) = V_n (x)$, for all $x>0$.
 In the same manner as
above, and using the function $V$,  define $V^\epsilon$ for all $\epsilon \in
(0,1)$.

Since $\lim_{n \to \infty} V_n(x) = V(x)$ and
$\lim_{n \to \infty} V'_n(x) = V'(x)$  \emph{uniformly} for $x \in
[\epsilon, 1]$, we have that $\lim_{n \to \infty} V_n^\epsilon(x) =
V^\epsilon(x)$, \emph{uniformly} for $x \in [0, 1]$.
Notice that this
uniform convergence fails in general for $\epsilon = 0$, unless
$(U_n)_{n \in \Natural}$ is uniformly bounded from above
(equivalently, if $(V_n)_{n \in \Natural}$ is uniformly bounded from
above). It follows that
\begin{equation}%
\label{equ:n-one}
    \begin{split}
\forall\, \epsilon \in (0,1)\,\
\exists\, n_1(\epsilon) \in \Natural,\ \forall\, x\in [0,1],\
\forall\, n\geq n_1(\epsilon),\
V_n^\epsilon(x) \geq V^\epsilon(x) - \epsilon.
    \end{split}
\end{equation}

\smallskip

Let $\tV_n$ denote the \emph{convex minor} of the family $\{ V_n, V_{n+1},
\ldots\}$; that is, $\tV_n$ is the largest convex function that is
dominated by all $V_k$ for $k \geq n$. Each $\tV_n$ is clearly
convex and decreasing. Observe also that $\tV_n(x) \geq - x$.
Indeed, remembering that $V_n(1) = V'_n(1) = -1$, for all
$n\in\Natural$, one concludes that
  $V_n(x) \geq - x$
 for all $n
\in \Natural$.
 In fact,
\[
-1 + \int_1^x \Big( \inf_{k \geq n} V'_k(u) \Big) \ud u \ \leq \
\tV_n(x) \ \leq \ V(x).
\]
This last expression, and the fact that $\lim_{n \to \infty} V'_n =
V'$ uniformly on compact subsets of $(0, \infty)$ imply that
$\lim_{n \to \infty} \uparrow \tV_n = V$ uniformly on compact
subsets of $(0, \infty)$.

\smallskip

Define now the ``average'' functions $\oV_n (x) :=
\tV_n(x) / x$, $x>0$ for all $n \in \Natural$, as well as $\oV (x) := V(x)
/ x$, $x>0$. Observe that $\lim_{n \to \infty} \uparrow \oV_n = \oV$
(increasing limit). The following, stronger, statement holds as well.
\begin{lem}
\label{lem:uniform} $\lim_{n \to \infty} \oV_n = \oV$, uniformly on
$[1,\infty)$.
\end{lem}
\proof We base the proof on Dini's theorem.  In order to be able to
use it we have to ensure that the sequence $\oV_n (\infty)$
increases and converges to $\oV (\infty)$, where $\oV_n(\infty):=
\lim_{x \to \infty} \oV_n (x)$ and $\oV(\infty) := \lim_{x \to
\infty} \oV (x) = \lim_{x \to \infty} V' (x) = 0$ (since $V$ is the
convex conjugate of $-U(- \cdot)$).

The normalization $V_n(1) = V'_n(1) = -1$ implies that $\tV_n(1) = \tV'_n(1) =
-1$ and $V(1) = V'(1) = -1$, and that $\oV$ and all $\oV_n$
are \emph{increasing} for $x \in [1, \infty)$.

Then, for an arbitrary $\delta > 0$, pick $M >
1$  so that $\oV(M) > - \delta / 2$ and $n_2
\equiv n_2 (\delta, M) \in \Natural$  so that $\oV_n(M)
> \oV (M) - \delta / 2 > - \delta$ for all $n \geq n_2$.
It follows that $\oV_n(\infty) \geq \oV_n (M)  > - \delta$ for all $n \geq n_2$
and thus that $\lim_{n \to \infty} \oV_n(\infty) = \oV(\infty) = 0$.
As proclaimed, Dini's theorem will imply that $\lim_{n \to \infty}
\oV_n = \oV$ \emph{uniformly} on $[1, \infty)$. \qed

\begin{lem} \label{lem: dual stuff}
The mapping $(z, y) \mapsto z V^\epsilon (y/z)$ is convex in $(z, y)
\in (0, \infty)^2$. Furthermore, for each $\epsilon>0$,
 there exists $n_0(\epsilon) \in
\Natural$ such that for all
$n \geq n_0(\epsilon)$ we have
\[
z V^\epsilon_n(y/z) \geq z V^\epsilon (y/z) - \epsilon (y + z).
\]
for all pairs $(z, y) \in (0, \infty)^2$.
\end{lem}

\begin{proof} The fact that $(z, y) \mapsto z V^\epsilon (y/z)$ is
convex in $(z, y) \in (0, \infty)^2$ is a consequence of the
convexity of $V^\epsilon$ and is quite standard. A detailed proof
can be found, for example, in \cite{urrlem}, page 90.

For the second claim, pick $\epsilon>0$, and
use Lemma \ref{lem:uniform} to find a natural number $n_3(\epsilon)$
such that
$\oV_n^\epsilon (x) \geq
\oV^\epsilon (x) - \epsilon$ for \emph{all} $x \geq 1$ and $n \geq
n_3 (\epsilon)$. Then, pick $n_1(\epsilon)$ as in (\ref{equ:n-one}).
Finally, choose $n_0 (\epsilon) := \max \{n_1(\epsilon),
n_3(\epsilon)\}$. For all $n \geq n_0 (\epsilon)$ we now have:
\[
z V^\epsilon_n(y/z) \geq z (V^\epsilon (y/z) - \epsilon ) \indic_{\{
y \leq z \}} + y (\oV (y/z) - \epsilon ) \indic_{\{ y > z \}} \geq z
V^\epsilon (y/z) - \epsilon (y + z).
\]
\end{proof}

\subsubsection{The conclusion of the proof of (\ref{eq: lower semicontinuity of dual})}
Let $\epsilon \in (0, 1)$ be fixed, but arbitrary. According to
Lemma \ref{lem: dual stuff}, for all $n \geq n_0(\epsilon)$ we have
$Z_n V_n^\epsilon (g_n / Z_n)  \geq Z_n V^\epsilon (g_n / Z_n) -
\epsilon (Z_n + g_n)$; applying expectation with respect to $\prob$
and taking limits (remember that we have passed in a subsequence so
that $\lim_{n \to \infty} v_n(y, r)$ exists) we get
\begin{equation} \label{eq: lower semicont estimate 1}
\lim_{n \to \infty} v_n(y, r) \geq \limsup_{n \to \infty} \expec
[Z_n V_n^\epsilon (g_n / Z_n)] \geq \limsup_{n \to \infty} \expec
[Z_n V^\epsilon (g_n / Z_n)] - (y + 1) \epsilon.
\end{equation}
Apply Lemma \ref{lem: dual stuff} again to get $\zeta_n V^\epsilon
(h_n / \zeta_n) \ \leq \ \sum_{k = n}^{m_n} \alpha^n_k Z_k
V^\epsilon (g_k / Z_k)$, where the sequence $(h_n)_{n \in \Natural}$
is the one of \eqref{eq: nice DS L0 convexity lemma}; this implies
that
\[
\limsup_{n \to \infty} \expec [Z_n V^\epsilon (g_n / Z_n)] \ \geq \
\limsup_{n \to \infty} \sum_{k = n}^{m_n} \alpha^n_k \expec [Z_k
V^\epsilon (g_k / Z_k)] \ \geq \ \limsup_{n \to \infty} \expec [
\zeta_n V^\epsilon (h_n / \zeta_n) ].
\]
A combination of this inequality
with the  estimate \eqref{eq: lower semicont
estimate 1} yields that
\begin{equation} \label{eq: lower semicont estimate 2}
\lim_{n \to \infty} v_n(y, r) \ \geq \ \limsup_{n \to \infty} \expec
[ \zeta_n V^\epsilon (h_n / \zeta_n) ] - (y + 1) \epsilon.
\end{equation}

\smallskip

Since $\oV^\epsilon$ is increasing on $[1, \infty)$ and satisfies
$\oV^\epsilon(1) = -1$ and $\oV^\epsilon(\infty) = 0$, one can
choose $M > 1$ such that $\oV^\epsilon (M) = - \epsilon$ and define
$\oV^{\epsilon, M}$ by requiring $\oV^{\epsilon, M} (x) =
\oV^\epsilon (x)$ for $0 < x \leq M$, $\oV^{\epsilon, M} (x) = 0$
for all $x \geq M+1$, and interpolating in a continuous way
between $M$ and $M+1$ so that $\oV^{\epsilon} \leq  \oV^{\epsilon,
M}$. Then $\oV^{\epsilon, M} - \epsilon \leq \oV^{\epsilon} \leq
\oV^{\epsilon, M}$
and
\begin{equation} \label{eq: lower semicont estimate 3}
\zeta_n V^\epsilon (h_n / \zeta_n) = h_n \oV^\epsilon (h_n /
\zeta_n) \geq h_n \oV^{\epsilon, M} (h_n / \zeta_n) - \epsilon h_n.
\end{equation}
Observe that $h_n \oV^{\epsilon, M} (h_n / \zeta_n) \leq
V^\epsilon(0) \zeta_n$; also, since $\oV^{\epsilon, M} (x) \geq -1$
for all $x > 0$ and $\oV^{\epsilon, M} (x) = 0$ for $x > M+1$, we
have
\[
h_n \oV^{\epsilon, M} (h_n / \zeta_n) = h_n \oV^{\epsilon, M} (h_n /
\zeta_n) \indic_{\{ h_n \leq (M+1) \zeta_n\}} \geq - (M+1) \zeta_n.
\]
It follows that $| h_n \oV^{\epsilon, M} (h_n / \zeta_n) | \leq
\kappa^{\epsilon, M} \zeta_n$, where $\kappa^{\epsilon, M} := \max
\{V^\epsilon(0) ,M+1 \}$. The sequence $(\zeta_n)_{n \in \Natural}$
is $\prob$-uniformly integrable with $\lim_{n \to \infty} \zeta_n =
1$ a.s, and  $\lim_{n \to \infty} h_n = h$, a.s. So,
by \eqref{eq: lower semicont estimate 3}, we have
\[
\limsup_{n \to \infty} \expec [\zeta_n V^\epsilon (h_n / \zeta_n)] \
\geq \ \expec [h \oV^{\epsilon, M} (h)] - \epsilon y \ \geq \ \expec
[h \oV^\epsilon (h)] - \epsilon y \ = \ \expec [V^\epsilon (h)] -
\epsilon y.
\]
Combining this last estimate with \eqref{eq: lower semicont estimate
2} we get
\[ \lim_{n \to \infty} v_n (y, r) \ \geq \ \limsup_{n \to \infty}
\expec [\zeta_n V^\epsilon (h_n / \zeta_n)] - (y + 1) \epsilon \
\geq \ \expec [V^\epsilon (h )] - (2y + 1) \epsilon
\]
Now, since $V^\epsilon (h ) \geq - h$ and $h \in \Lb^1$, one can use
the monotone convergence theorem in the last inequality and the fact
that $h \in \D(y, r)$ to get (as $\epsilon \downarrow 0$) that
\[
\lim_{n \to \infty} v_n (y, r) \ \geq \ \expec^\prob [V (h )] \ \geq
\ v(y,r),
\]
which finishes the proof.

\subsection{Limiting behavior of the sequence of  dual value
  functions}
\label{subsec: continuity of dual}

From now on, we assume that all four conditions \eqref{ass: NFLVR},
\eqref{ass: S-REP}, \eqref{ass: CONV} and \eqref{ass: UI} hold. The first
order of business is
 to study the behavior of the limit superior of the
sequence of the dual value functions. Then, we combine the obtained
result with that of
subsection \ref{subsec: lower semicontinuity of dual}.

\subsubsection{Auxiliary Results} For future reference,
for any $p \in \prices$, $\qprob \in \Q'(p)$ and $(y, r) \in \eL$
such that $y p = r$ we define
\begin{equation} \label{eq: bounded from below dual elements}
\B (y, r, \qprob) := \{ g \in \D (y,r) \ | \ \tfrac{1}{g} \ud \qprob /
\ud \prob\in\Lb^{\infty}. \}.
\end{equation}
Since $\D (y,r)$ is
convex and $y \ud \qprob / \ud \prob \in \D (y,r)$, we have that
for all $g \in \D (y,r)$ and $k \in
\Natural$, $k^{-1} (y \ud \qprob / \ud \prob) + (1- k^{-1}) g \in \B
(y,r, \qprob)$; in particular, $\B (y,r, \qprob) \neq \emptyset$.

\begin{lem} \label{lem: dual is inf over bounded below dual elements}
Fix $y > 0$ and $p \in \prices$, and  let $\qprob \in \Q'(p)$ be such
that $V^+ ( y \ud \qprob / \ud \prob ) \in \Lb^1 (\prob)$. Then,
with $r := y p$ and $\B (y, r, \qprob)$ defined in \eqref{eq:
bounded from below dual elements}, we have
\[
v (y, r) = \inf_{g \in \B (y, r, \qprob)} \expec [ V (g)].
\]
\end{lem}

\proof Let $g_* \in \D (y,r)$ satisfy $v (y, r) = \expec [V (g_*)]$.
For all $k \in \Natural$ define $g_{*, k} := (1 - k^{-1}) g_* +
k^{-1} (y \ud \qprob / \ud \prob)$. Then,  $g_{*, k} \in \B (y, r,
\qprob)$ and $\expec [ V (g_{*, k})] \leq (1 - k^{-1}) \expec [V
(g_*)] + k^{-1} \expec [V ( y \ud \qprob / \ud \prob )]$. Finally,
since $V( y \ud \qprob / \ud \prob ) \in
\Lb^1$, we get that
$\lim_{k \to
\infty} \expec [ V (g_{*, k})] = v (y, r)$. \qed

\begin{lem} \label{lem: continuity in L^1 for bounded below dual elements}
Suppose that for some $f \in \Lb_+^0$ the collection $(Z_n V^+_n(f/
Z_n))_{n \in \Natural}$ of random variables is $\prob$-uniformly
integrable. Let also $g \in \Lb^1(\prob)$ be such that $g \geq f$
almost surely. Then, $\lim_{n \to \infty} Z_n V_n (g/ Z_n) = V (g)$
in $\Lb^1 (\prob)$.
\end{lem}

\proof Since we have $\lim_{n \to \infty} Z_n V_n (g/ Z_n) = V(g)$
in $\Lb^0$, we only have to show that the collection $(Z_n V_n (g/
Z_n))_{n \in \Natural}$ of random variables is $\prob$-uniformly
integrable.

Each $V_n$ is decreasing, thus $Z_n V^+_n (g/ Z_n) \leq Z_n V^+_n(f/
Z_n)$, and $(Z_n V^+_n(f/ Z_n))_{n \in \Natural}$ is
$\prob$-uniformly integrable by assumption.

On the other hand, since $V_n(x) \geq - x$ for all $n \in \Natural$
we get $Z_n V^-_n (g/ Z_n) \leq g$. The uniform integrability of
$(Z_n V^-_n (g/ Z_n))_{n \in \Natural}$ now follows from the fact that
 $g \in \Lb^1(\prob)$.\qed

 \subsubsection{An upper semicontinuity-property of the sequence of
   the dual value functions}

We proceed here to show that for fixed $(y, r) \in \eL$ we have
\begin{equation} \label{eq: upper semicontinuity of dual}
\limsup_{n \to \infty} v_n(y, r) \leq v(y, r)
\end{equation}

With $p := y^{-1} r$, pick some $\qprob \in \Q'(p)$ such that $(Z_n
V^+_n( y \ud \qprob / \ud \prob_n))_{n \in \Natural}$ is
$\prob$-uniformly integrable (observe that
\emph{this} is where we use our \eqref{ass: UI} assumption).
Then, $V^+( y \ud \qprob / \ud \prob ) \in \Lb^1 (\prob)$, and,
according to Lemma \ref{lem: dual is inf over bounded below dual
elements}, $v(y, r) = \inf_{g \in \B (y, r, \qprob)} \expec [V
(g)]$.

\smallskip

For any $g \in \B (y, r, \qprob) \subseteq \D (y, r)$, Lemma
\ref{lem: continuity in L^1 for bounded below dual elements} implies
that $\lim_{n \to \infty} \expec [Z_n V_n (g/ Z_n)] = \expec [V
(g)]$; it follows that $\limsup_{n \to \infty} v_n(y_n, r_n) \leq
\limsup_{n \to \infty} \expec [Z_n V_n (g / Z_n)] = \expec [ V (g
)]$. Taking the infimum over all $g \in \B (y, r, \qprob)$ in the
right-hand-side of the last inequality we arrive at \eqref{eq: upper
semicontinuity of dual}.

\subsection{Limits of  sequences of primal value functions and
  marginal utility-based prices}
\label{subsec: continuity of primal and util indif prices}

In   subsection \ref{subsec: continuity of dual} above, we
established that $(v_n)_{n \in \Natural}$ converges pointwise to $v$
on $\eL$. Since all the functions involved are convex, the
convergence is uniform  on compact subsets of $\eL$. Thanks to
the strong stability properties of the family of convex functions on
finite-dimensional spaces, this fact (and this fact only) yields
convergence of the concave primal value functions, as well as the
related sub-differentials and super-differentials to the
corresponding limits. Indeed, by  Theorem 7.17, p.~252 in
\cite{RocWet98},  pointwise convergence on the interior of the
effective domain of the limiting function is equivalent to the
weaker notion of {\em epi-convergence}. In our case, primal value
functions are all defined on $\K$ and the dual value functions on
$\eL$, both of which are, in fact, open thanks to the assumption
(\ref{ass: N-TRAD}). For reader's convenience, we repeat the
definition of epi-convergence
\begin{defn}[Definition 7.1., p.~240, \cite{RocWet98}]
\label{def:epi} Let $\seq{f}$ be a sequence of lower semicontinuous
and proper convex functions defined on some Euclidean space
$\Real^d$. We say that $f_n$ {\em epi-converges} to $f$ --- and
write $f_n \toe f$ if
\begin{enumerate}
\item $\forall\, x\in\Real^d,\ \forall\, x_n\to x$, $\liminf f_n(x_n)\geq
  f(x)$;
\item $\forall\, x\in\Real^d,\ \exists\, x_n\to x$, $\limsup f_n(x_n)\leq
  f(x)$.
\end{enumerate}
\end{defn}
Epi-convergence seems to be tailor-made to interact well with
conjugation. Denoting the convex conjugation by $(\cdot)^*$,
Theorem 11.34., p.~500., in \cite{RocWet98} states that
\begin{equation}
    \nonumber
    \begin{split}
f_n\toe f\ \Leftrightarrow\ f_n^*\toe f^*,
    \end{split}
\end{equation}
as long as the functions $\seq{f}$ are proper and
lower-semicontinuous. An immediate consequence  of the above fact is
that
\begin{equation}
    \nonumber
    \begin{split}
u(x,q)=\lim_n u_n(x,q),\text{ for all }(x,q)\in\K.
    \end{split}
\end{equation}

Moreover, the functions $u_n$ as well as the limiting function $u$ are
know to be convex, so the stated pointwise convergence is, in fact,
uniform on compacts. Therefore, the following, stronger, conclusion
holds
\begin{equation}
    \nonumber
    \begin{split}
u(x_n,q_n)=\lim_n u_n(x_n,q_n),\text{ for all }(x_n,q_n)\to
(x,q)\in\K.
    \end{split}
\end{equation}

The list of pleasant properties of epi-convergence is not exhausted
yet. By Theorem 12.35., p.~551 in \cite{RocWet98}, epi-convergence of
convex functions implies the
convergence of their sub-differentials, in the
sense of graphical convergence, as  defined  below (the dimension
$d\geq 1$ of the underlying space is general, but will be applied as
$d=N+1$):
\begin{defn}[Definition 5.32., p.~166, Proposition 5.33., p.~167]
  Let $T,\seq{T}:\Real^d\rightrightarrows \Real^d$ be a sequence of
  correspondences.  We say that $T_n$ {\em graphically converges} to
  $T$, and write $T_n\tog T$, if for all $x\in\Real^d$,
\begin{equation}
    \nonumber
    \begin{split}
      \bigcup_{\set{x_n \to x}} \limsup_n T_n(x_n) \subseteq T(x)
      \subseteq \bigcup_{\set{x_n\to x}} \liminf_n T_n(x_n),
    \end{split}
\end{equation}
where $\liminf$ and $\limsup$ should be interpreted in the usual
set-theoretical sense, and the unions are taken over all sequences
$\seq{x}$ in $\Real^d$, converging to $x$.
\end{defn}

In order to combine the results mentioned above and
illustrate the notion of
graphical convergence in more familiar terms,
we state and prove the following simple observation.
\begin{prop}
\label{prop:last}
Suppose that $f_n\toe f$, and let $\seq{x}$ be a sequence in $\Real^d$
converging towards some $x\in\Real^d$. Then
\begin{equation}%
\label{equ:limsup}
    \begin{split}
\limsup_n \partial f_n(x_n) \subseteq
\partial f(x).
    \end{split}
\end{equation}
Further, let $\sup\{\norm{y}\,|\,y\in \partial f_n(x_n),\
n\in\Natural\}<\infty$.  Then, for each $\eps>0$ there exists
$n(\eps)\in\Natural$ such that for $n\geq n(\eps)$, $\partial f_n(x_n)
\subseteq \partial f(x)+\eps B$, with $B$ denoting the unit ball of
$\Real^d$.
\end{prop}

\begin{proof}
  The first statement follows from the definition of graphical
  convergence, and its relationship to epi-convergence. For the
  second, suppose, to the contrary, that we can find an $\eps>0$ and
  an increasing sequence $n_k\in\Natural$ such that there exist
  points $x^*_k\in \partial f_{n_k}(x_{n_k})$ such that
  $d(x^*_{k},\partial f(x))>\eps$. If $(x^*_k)_{k\in\Natural}$ has a convergent
  subsequence, then its limit $x^*_0$ has to satisfy $d(x^*_0,\partial
  f(x))\geq \eps$ --- a contradiction with (\ref{equ:limsup}).
  Therefore, there exists a subsequence of $(x^*_k)_{k\in\Natural}$, converging to
  $+\infty$ in norm. This, however, contradicts
  assumed uniform boundedness of subdifferentials.
\end{proof}

The sequences $\seq{v}$ and $\seq{u}$ converge in a pointwise
fashion, uniformly on compacts. It follows now directly from
Definition \ref{def:epi} that $\seq{u}$ and $\seq{v}$ converge in
the epi sense towards $u$ and $v$.
 As we have already mentioned above, epi convergence
implies graphical convergence of the subdifferentials.
To be able to use the additional conclusion of
Proposition \ref{prop:last}, we need to establish
uniform boundedness of the superdifferentials of the functions
$u(x_n,q_n;U_n,\prob_n)$, when $(x_n,q_n)$ live in a compact subset of
$\K$. By Theorem \ref{thm: hugkra}, these are all of the form
\begin{equation}
    \nonumber
    \begin{split}
\partial u(x_n,q_n; U_n, \prob_n)=\{y_n\}\times{
  \partial_{q} u(x_n,q_n;U_n, \prob_n)},\text{ where }
y_n=\frac{\partial}{\partial x} u(x_n,q_n;U_n, \prob_n).
    \end{split}
\end{equation}
It is an easy consequence of the second inclusion in the definition
of the graphical convergence, and the differentiability in the
$x$-direction of all functions $u$, $\seq{u}$ that $y_n\to
y=\tfrac{\partial}{\partial x} u(x,q;U, \prob)$. In particular, the
sequence $\seq{y}$ is bounded away from zero, so, in order to use
Proposition \ref{prop:last},
 it is enough to show that the
sets $\prices(x_n, q_n; U_n, \prob_n)$ of utility-based prices are uniformly
bounded. This fact follows immediately, once we recall that those
are always contained in the sets of arbitrage-free
prices, which are uniformly bounded by (S-REP). It remains to use
Proposition \ref{prop:last}
 above and remember the characterization (\ref{eq: util indif
   prices}), to complete the proof of parts (1) and (3) of our main
 Theorem \ref{thm: main}.

\subsection{Continuity of the optimal dual element and optimal wealth processes} \label{subsec: continuity of optimal wealth}

We conclude  the proof of Theorem \ref{thm: main}, tackling item
(2) on convergence of the optimal terminal wealth and the optimal
dual elements. Let  $(x_n,
q_n)_{n \in \Natural}$ with $\lim_{n \to \infty} (x_n, q_n) =: (x,
q) \in \K$ and $(y_n, r_n)_{n \in
\Natural}$ with $\lim_{n \to \infty} (y_n, r_n) =: (y, r) \in \eL$ be,
respectively,
a $\K$-valued and an $\eL$-valued sequence.

\subsubsection{Preparation}

Remember from Theorem \ref{thm: hugkra} that the optimal dual and
optimal primal elements are connected via
\[
\hat{X}_T (x_n, q_n; U_n, \prob_n) + \langle q_n, f \rangle = - V'_n
(\hat{Y}_T (y_n, r_n; V_n, \prob_n)), \textrm{ where } (y_n, r_n)
\in \partial u(x_n, q_n).
\]
If we show that $\lim_{n \to \infty} \hat{Y}_T (y_n, r_n; V_n,
\prob_n) = \hat{Y}_T (y, r; V, \prob)$ in $\Lb^0$ for all sequences
$(y_n, r_n)_{n \in \Natural}$ that are $\eL$-valued with $\lim_{n
\to \infty} (y_n, r_n) =: (y, r) \in \eL$, then the convergence of
the random variables $\hat{X}_T (x_n, q_n; U_n, \prob_n)$ to
$\hat{X}_T (x, q; U, \prob)$ in $\Lb^0$ will follow as well. Indeed,
fix a $\K$-valued sequence $(x_n, q_n)_{n \in \Natural}$ with
$\lim_{n \to \infty} (x_n, q_n) =: (x, q) \in \K$; from the upper
hemicontinuity property proved in Subsection \ref{subsec: continuity
of primal and util indif prices}, we can choose for each $n \in
\Natural$ some $(y_n, r_n) \in
\partial u(x_n, q_n)$ in such a way as to have $\lim_{n \to \infty}
(y_n, r_n) =: (y, r) \in
\partial u(x, q) \subseteq \eL$. The claim now follows easily; indeed,
$\lim_{n \to \infty} \hat{Y}_T (y_n, r_n; V_n, \prob_n) = \hat{Y}_T
(y, r; V, \prob)$ in $\Lb^0$, $\hat{Y}_T (y, r; V, \prob) >
0$, a.s., and  $(V'_n)_{n \in \Natural}$ converges uniformly
to $V'$ on compact subsets of $(0, \infty)$.

To ease notation we write $g_n := Z_n \hat{Y}_T (y_n, r_n; V_n,
\prob_n)$ and $g := \hat{Y}_T (y, r; V, \prob)$. Then $g_n \in
\D(y_n, r_n)$ satisfies $v_n(y_n, r_n) = \expec [Z_n V_n(g_n /
Z_n)]$ for each $n \in \Natural$ and $g \in \D (y, r)$ satisfies
$v(y, r) = \expec [V(g)]$. The condition $\lim_{n \to \infty}
\hat{Y}_T (y_n, r_n; V_n, \prob_n) = \hat{Y}_T (y, r; V, \prob)$ in
$\Lb^0$ that we need to prove translates to $\lim_{n \to \infty}
(g_n / Z_n) = g$ in $\Lb^0$.  Assume that an arbitrary subsequence
(whose indices are not relabeled) has already been extracted from
$(g_n / Z_n)_{n\in\Natural}$. It suffices to show that $\lim_{k
\to \infty} (g_{n_k} / Z_{n_k}) = g$ in $\Lb^0$ along some further
\emph{subsequence} $(g_{n_k} / Z_{n_k})_{k \in \Natural}$.

We now make a further step that --- although seemingly confusing ---
will prove useful in the sequel of our proof. With $p := y^{-1} r$,
we pick some $\qprob \in \Q(p)$ such that $f := y \ud \qprob / \ud
\prob$ satisfies $V^+ ( f ) \in \Lb^1 (\prob)$. For each $n \in
\Natural$ let $f_n := n^{-1} f + (1 - n^{-1}) g$. Note that $f_n
\in \B(y,r,\qprob)$ for all $n \in \Natural$, with  $\B(y,r,\qprob)$
as in \eqref{eq: bounded from below dual elements},  and that
$\lim_{n \to \infty} f_n = g$ in $\Lb^0$.

For any $m
\in \Natural$ define
\begin{equation} \label{eq: C_m}
C_m := \{ (a, b) \in \Real^2 \ | \ 1/m \leq a \leq m, \ 1/m \leq b
\leq m, \textrm{ and } |a - b| > 1 / m \}.
\end{equation}
Combining the discussion above with the facts that $\lim_{n \to
\infty} Z_n = 1$ in $\Lb^0$ and that both sequences $(g_n)_{n \in
\Natural}$ and $(f_n)_{n \in \Natural}$ are bounded in $\Lb^0$, we
conclude that in order to prove
that $ \lim_{n \to \infty} (g_n / Z_n) = g\text{ in
  $\Lb^0$}$, we need to establish
the following claim:
\begin{claim}
\label{cla:only}
There exist a strictly increasing
sequence $\{n_k\}_{k\in\Natural}$ so that the subsequences
$(g_{n_m})_{m \in \Natural}$ and $(f_{n_m})_{m \in \Natural}$ of
$(g_n)_{n \in \Natural}$ and $(f_n)_{n \in \Natural}$ respectively, satisfy
\begin{equation} \label{eq: conergence of optimal dual}
\lim_{m \to \infty} \prob_{n_m} \left[\left(
g_{n_m}/Z_{n_m},f_{n_m}/Z_{n_m}\right) \in C_m \right] = 0.
\end{equation}
\end{claim}
The above clarifies the reason why the sets $C_m$, $m \in \Natural$ of \eqref{eq: C_m} were introduced; in fact, this trick is a more elaborate version of the method used in the proof of Lemma A.1 of \cite{delsch94}.
\begin{rem}
For a sequence $(A_n)_{n \in \Natural}$ of $\F$-measurable sets,
$\lim_{n \to \infty} \prob_n [A_n] = 0$ is equivalent to $\lim_{n
\to \infty} Z_n \indic_{A_n} = 0$ in $\Lb^0$ (combining the
$\Lb^1(\prob)$-convergence of the last sequence with $\prob$-uniform
integrability of $(\prob_n)_{n \in \Natural}$) which, in view of the
fact $\lim_{n \to \infty} Z_n = 1$ in $\Lb^0$,  is equivalent to
$\lim_{n \to \infty} \indic_{A_n} = 0$ in $\Lb^0$, or in other words
that $\lim_{n \to \infty} \prob [A_n] = 0$. This justifies the use
of ``$\prob_{n_m}$'' instead of ``$\prob$'' in \eqref{eq: conergence
of optimal dual}.
\end{rem}

\subsubsection{Proof of Claim \ref{cla:only}}

For any $m \in \Natural$, the strict convexity of $V$ implies the
existence of some $\beta_m > 0$ such that for all $(a, b) \in (0,
\infty)^2$ we have
\[
V \Big( \frac{a + b}{2} \Big) \leq  \frac{V(a) + V(b)}{2} - \beta_m
\indic_{C_m} (a, b),\text{ for the set $C_m$ of \eqref{eq: C_m}.
}
\]
\emph{Uniform}
convergence of $(V_n)_{n \in \Natural}$ to $V$ on compact subsets of
$(0, \infty)$  implies that (with a possible lower, but still
strictly positive, choice of $\beta_m$) we still have
\[
V_n \Big( \frac{a + b}{2} \Big) \leq  \frac{V_n(a) + V_n(b)}{2} -
\beta_m \indic_{C_m} (a, b),
\]
for all $n \in \Natural$ and $(a, b) \in (0, \infty)^2$. Setting $a =
g_n / Z_n$, $b = f_k / Z_n$, multiplying both sides of the previous
inequality with $Z_n$, and taking expectation with respect to $\prob$,
one gets
\begin{equation}
    \nonumber
    \begin{split}
      \beta_m \prob_n \Big[ \Big(\frac{g_n}{Z_n},\frac{f_k}{Z_n} \Big)
      \in C_m \Big] &\leq \frac{1}{2} \expec \Big[ Z_n V_n
      \Big(\frac{g_n}{Z_n} \Big) \Big] + \frac{1}{2} \expec \Big[ Z_n
      V_n \Big(\frac{f_k}{Z_n} \Big) \Big] -
      \expec \Big[ Z_n V_n \Big(\frac{g_n + f_k}{2 Z_n} \Big) \Big]\\
      &\leq \frac{1}{2} v_n (y_n, r_n) + \frac{1}{2} \expec \Big[ Z_n
      V_n \Big(\frac{f_k}{Z_n} \Big) \Big] - v_n \Big( \frac{y_n +
        y}{2}, \frac{r_n + r}{2} \Big),
    \end{split}
\end{equation}
where for the third term of the last inequality we have used that
fact that
\[
\frac{g_n + f_k}{2 Z_n} \ \in \ \D \Big( \frac{y_n + y}{2},
\frac{r_n + r}{2} \Big).
\]
Invoking Lemma \ref{lem: continuity in L^1 for bounded below dual
  elements}, we know that $\lim_{n \to \infty} \expec [Z_n V_n(f_k /
Z_n)] = \expec [V(f_k)]$ for all fixed $k \in \Natural$.  Furthermore,
the proof of Lemma \ref{lem: dual is inf over bounded below dual
  elements} shows that $\lim_{k \to \infty} \expec [V(f_k)] = \expec
[V(g)] = v(y, r)$. Using also the uniform convergence (on compact
subsets of $\eL$) of $(v_n)_{n \in \Natural}$ to $v$, we see that
we can choose $k_m$ and $n_m$ large enough so that
$\prob_{n_m} [ Z_{n_m}^{-1} ({g_{n_m}},f_{k_m}) \in C_m] \leq 1/m$. It
is a matter of subsequence manipulation to show that one can,
in fact, choose a universal strictly increasing sequence $n_m=k_m$,
$m\in\Natural$ with all the desired properties.  This proves
\eqref{eq: conergence of optimal dual}, and concludes the proof of our
main Theorem \ref{thm: main}.

\bibliographystyle{siam}
\bibliography{stability}
\end{document}